\documentclass{llncs} 
\usepackage[english]{babel}
\usepackage{amsfonts,verbatim}

\usepackage[pagewise,displaymath,mathlines]{lineno}

\usepackage[usenames]{color}

\title{On strong alt-induced codes} 

\author{Ngo Thi Hien}
\institute{Hanoi University of Science and Technology\\ Email: hien.ngothi@hust.edu.vn}

\date{}

\begin{document}

\maketitle

\renewcommand\refname{\normalsize \centerline{ REFERENCES}}
\pagestyle{plain}
\pagestyle{myheadings}
\markboth{Ngo Thi Hien}{On strong alt-induced codes}

\begin{abstract}
Strong alt-induced codes, a particular case of alt-induced codes, has been introduced and considered by D. L. Van and the author in earlier papers.
In this note, an algorithm to check whether a regular code is strong alt-induced or not is proposed, and the embedding problem for the classes of prefix (suffix, bifix) strong alt-induced codes in both the finite and regular case is also exhibited.



\end{abstract}

\section{Introduction}
The theory of length-variable codes has now become a part of theoretical computer science and of formal languages. A code is a language such that every text encoded by words of the language can be decoded in a unique way or, in other words, every coded message admits only one factorization into code-words. Codes are useful in many areas of application such as information processing, data compression, cryptography, information transmission and so on. For background of the theory of codes we refer to \cite{BP85,JK97,Shyr91}.

Alternative codes, an extension of the notion of ordinary codes, have been first introduced and considered by P. T. Huy et al. in 2004 \cite{HN2004}.
The authors demonstrated the importance characterizations as well as algorithms to test for alternative codes and their subclasses \cite{HH2012,Hf2016,Hf2017,HfV2017,HN2004,VNH2010}.
An alternative code is nothing but a pair $(X, Y)$ of languages such that $XY$ is a code and the product $XY$ is unambiguous \cite{HN2004}.
We say that the code $XY$ is induced by the alternative code $(X,Y)$. A code is said to be a {\it code induced by alternative code} (an {\it alt-induced code}, for short) if there exists an alternative code which induces it.

Strong alt-induced codes, a particular case of alt-induced codes which is, has been introduced and considered in \cite{HfV2017}. Several interesting characteristic properties of such codes were established.

As well-known, for many classes of codes, a slight application of the Zorn's lemma shows that every code in a class is included in a code maximal in the same class (not necessarily maximal as a code). For thin codes, regular codes in particular, the maximality is equivalent to the completeness, which concerns with optimal use of transmission alphabet. Thus maximal codes are important in both theoretical and practical points of view.

For a given class $C$ of codes, a natural question is whether every code $X$ satisfying some property $\mathfrak{p}$ (usually, the finiteness or the regularity) is included in a code $Y$ maximal in $C$ which still has the property $\mathfrak{p}$. This problem, which we call the {\it embedding problem} for the class $C$, attracts a lot of attention. Unfortunately, this problem was solved only for several cases by means of different combinatorial techniques (see \cite{BrP1999,HHV2004,VH2011} and the papers cited there). 

Our aim in this note is to propose an algorithm to check whether a regular code is strong alt-induced or not, as well as solving the embedding problem for the classes of prefix (suffix, bifix) strong alt-induced codes in both the finite and regular case.

\section{Preliminaries}
Let $A$ throughout be a finite alphabet, i.e. a non-empty finite set of symbols, which are called letters. Let $A^*$ be the set of all finite words over $A$. The empty word is denoted by $\varepsilon$ and $A^{+}$ stands for $A^{*} \setminus \{\varepsilon\}$. The number of all the occurrences of letters in a word $u$ is the {\it length} of $u$, denoted by $|u|$. Any subset of $A^*$ is a {\it language} over $A$.
A language $X$ is a {\it code} over $A$ if any word  in $A^*$ has at most one factorization into words of $X$. A code $X$ is maximal over $A$ if it is not properly contained in any other code over $A$. Let $C$ be a class of codes over $A$ and $X \in C$. We say that the code $X$ is maximal
in $C$ if it is not properly contained in any code in $C$.

	A word $u$ is called an {\it infix} (a {\it prefix}, a {\it suffix}) of a word $v$ if there exist words $x, y$ such that $v=xuy$ (resp., $v = uy$, $v= xu$). The infix (prefix, suffix) is {\it proper} if $xy \ne \varepsilon$ (resp., $y \ne \varepsilon$, $x \ne \varepsilon$).
A word $u$ is a {\it subword} of a word $v$ if, for some $n \geq 1, u = u_1  \dots  u_n, v = x_0u_1x_1  \dots  u_nx_n $  with  $u_1,  \dots, u_n, x_0,  \dots , x_n \in A^*$. If  $x_0 \dots x_n \ne \varepsilon$ then $u$ is called a {\it proper subword} of $v$.
The set of proper prefixes of a word $w$ is denoted by ${\rm Pref}(w)$. We denote by ${\rm Pref}(X)$ the set of all proper prefixes of the words in $X \subseteq A^*$. The notations ${\rm Suff} (w)$ and ${\rm Suff} (X)$ are defined in a similar way.

\begin{definition} Let $X$ be a non-empty subset of $A^+$.
\begin{enumerate}
\renewcommand{\labelenumi}{{\rm (\roman{enumi})}}
	\item $X$ is a {\it prefix} ({\it suffix}) {\it code} if no word in $X$ is a proper prefix (resp., suffix) of another word in $X$, and $X$ is a {\it bifix code} if it is both a prefix code and a suffix code;
	\item $X$ is an {\it infix} (a {\it p-infix}, a {\it s-infix}) {\it code} if no word in $X$ is an infix of a proper infix (resp., prefix, suffix) of another word in $X$;
	\item $X$ is a {\it subinfix} ({\it p-subinfix, s-subinfix}) {\it code} if no word in $X$ is a subword of a proper infix (resp., prefix, suffix) of another word in $X$;
	\item $X$ is a {\it hypercode} if no word in $X$ is a proper subword of another word in it.
\end{enumerate}
\end{definition}

Prefix codes and their subclasses play a fundamental role in the theory of codes (see \cite{BP85,BrP1999,HHV2004,HV2006,JK97,Shyr91,VH2011}).

	For $X, Y \subseteq A^*$, the {\it product} of $X$ and $Y$ is the set $XY = \{xy \ | \ x \in X, y \in Y\}$. The product is said to be {\it unambiguous} if, for each $z \in XY$, there exists exactly one pair $(x,y) \in X \times Y$ such that $z = xy$.

	For $w \in A^*$, we define
\[w^{-1}X = \{u \in A^* \ | \ wu \in X\}, \ Xw^{-1} = \{u \in A^* \ | \ uw \in X\}.\]
These notations are extended to sets in a natural way:
\[X^{-1}Y = \bigcup_{x \in X}x^{-1}Y, \ XY^{-1} = \bigcup_{y \in Y}Xy^{-1}.\]

Let $(X,Y)$ be a pair of non-empty subsets of $A^+$, and let $u_1, u_2, \dots u_n \in X \cup Y, n \geq 2$. We say that $u_1u_2 \dots u_n$ is an {\it alternative factorization on $(X, Y)$} if $u_i \in X$ implies $u_{i+1} \in Y$ and $u_i \in Y$ implies $u_{i+1} \in X$ for all $i = 1, 2, \dots, n-1$.
Two alternative factorizations $u_1u_2 \dots u_n$ and $v_1v_2 \dots v_m$ on $(X, Y)$ are said to be {\it similar} if they both begin and end with words in the same set $X$ or $Y$.

\begin{definition} \label{D:sicode}
	Let $X$ and $Y$ be two non-empty subsets of $A^+$. The pair $(X, Y)$ is called an {\it alternative code} if no word in $A^+$ admits two different similar alternative factorizations on $(X, Y)$.
An alternative code $(X, Y)$ is called a {\it strong alternative code} if it satisfies the following conditions
		\[X^{-1}(XY) \subseteq Y, (XY)Y^{-1} \subseteq X.\]
\end{definition}

For more details of alternative codes and their subclasses we refer to \cite{Hf2016,HN2004,VNH2010}.

\begin{definition} \label{D:icode}
	A subset $Z$ of $A^+$ is called a {\it  code induced by an alternative code} ({\it alt-induced code}, for short) if there is an alternative code $(X, Y)$ over $A$ such that $Z = XY$.
An alt-induced code $Z$ is called {\it strong} if there exists a strong alternative code $(X,Y)$ generating it, $Z=XY$.
\end{definition}

\medskip
Now we formulate, in the form of lemmas, several facts which will be useful in the sequel.


\begin{lemma}[\cite{HHV2004}] \label{L:Cph-Clocat} 
The product of two p-infix (s-infix, infix, p-subinfix, s-subinfix, subinfix, hyper) codes is a  p-infix (resp., s-infix, infix, p-subinfix, s-subinfix, subinfix, hyper) code.
\end{lemma}

A subset $X$ of $A^*$ is {\it thin } if there exists at least one word $w \in A^*$ which is not an infix of any word in $X$, i.e. $X \cap A^*wA^* = \emptyset$.

\begin{lemma}[{\cite{BP85}}, page 69] \label{L:ThinReg}
	Any regular code is thin.
\end{lemma}
 
	Concerning the maximality of thin codes we have

\begin{lemma}[{\cite{BP85}}, Proposition 2.1, page 145] \label{L:ThinMb}
	Let $X$ be a thin subset of $A^+$. Then, $X$ is a maximal bifix code if and only if $X$ is both a maximal prefix code and a maximal suffix code.
\end{lemma}

The following results are characterizations for strong alt-induced codes and their subclasses.

\begin{lemma} [\cite{HfV2017}] \label{L:siaCode-Char}
	Let $X$ and $Y$ be non-empty subsets of $A^+$. Then, $XY$ is a strong alt-induced code if and only if $X$ is a prefix code, $Y$ is a suffix code and $XY$ is a code.
\end{lemma}

\begin{lemma} [\cite{HfV2017}] \label{L:siCpsb-Char}
	Let $X$ and $Y$ be non-empty subsets of $A^+$.
\begin{enumerate}
\renewcommand{\labelenumi} {\rm(\roman{enumi})} 
	\item $XY$ is a prefix (maximal prefix) strong alt-induced code if and only if $X$ is a prefix (resp.,  maximal prefix) code and $Y$ is a bifix code (resp., $Y$ is both a maximal prefix code and a bifix code);
	\item $XY$ is a suffix (maximal suffix) strong alt-induced code if and only if $X$ is a bifix code (resp., $X$ is both a maximal suffix code and a bifix code) and $Y$ is a suffix (resp., maximal suffix) code;
	\item $XY$ is a bifix (maximal bifix thin) strong alt-induced code if and only if $X$ and $Y$ are bifix (resp., maximal bifix thin) codes.
\end{enumerate}
\end{lemma}

	The following result concerns the maximality of prefix (suffix, bifix) codes. As usual, when $X$ is finite, by $\max X$ we denote the maximal wordlength of $X$.

\begin{lemma} [See \cite{BP85,BrP1999,VH2011}] \label{L:EP-Cpsb} The following holds true
\begin{enumerate}
\renewcommand{\labelenumi} {\rm(\roman{enumi})} 
	\item Any regular prefix (suffix, bifix) code is contained in a maximal one.
	\item Any finite prefix (suffix) code $X$ is contained in a maximal finite prefix (resp. suffix) code $Y$ with $\max Y = \max X$.
	\item Any finite bifix code is contained in a regular maximal bifix code.
\end{enumerate}
\end{lemma}

Let $X$ be a bifix code. We define the {\it indicator} of $X$ as the following function from $A^*$
into the set of integers:
$$L_X(w) = 1 + |w| - F_X(w),$$
where $F_X(w)$ is the number of occurrences of words of $X$ as infixes of $w$ and $|w|$ is the length of $w$.
An {\it interpretation} of a word $w$ is a triple $(s,x, p)$ where
$w = sxp$ and where $s \in (A^+)^{-1}X, x \in X^*, p \in X(A^+)^{-1}$. A {\it point} in the word $w$ is a pair $(u,v)$
such that $w = uv$ with $u,v \in A^*$. The interpretation $(s, x, p)$ is said to {\it pass} by the point
$(u, v)$ if there exist $y, z \in X^*$ such that $u = sy, v = zp$ and $x = yz$.
These notions can be extended to infinite words $w \in A^\omega$. An {\it interpretation} of $w$ is
a pair $(s,x)$ such that $w = sx, s \in A^-X$ and $x \in X^\omega$. The interpretation $(s,x)$ {\it  passes} by the point $(u, v)$ (with $u \in A^*, v \in A^\omega$) if there exist $y \in X^*, z \in X^\omega$ such that $u = sy, v=z$ and $x= yz$.

A (finite or infinite) word $w$ is called {\it full} if by any point of $w$ passes an interpretation. The set $X$ is called {\it sufficient} if the set of full words is infinite and {\it insufficient} otherwise.

A finite bifix code $X$ is called to be {\it nice} if either $X$ is insufficient or $X$ is sufficient and all infinite full words have the same number $n$ of interpretations.

\begin{lemma} [\cite{BrP1999}] \label{L:EP-Cfb}
Let $X$ be a finite bifix code.
\begin{enumerate}
\renewcommand{\labelenumi} {\rm(\roman{enumi})}
	\item If $X$ is insufficient, then for every $n \geq \max\{L_X(x) \ | \ x \in X\}$, $X$ is contained in a finite maximal bifix code of degree $n$.
	\item If $X$ is sufficient, two cases arise. Either there exist two infinite full words with a
different number of interpretations. Then $X$ is not contained in any finite maximal bifix code. Or all infinite full words have the same number $n$ of interpretations. The possible finite maximal bifix codes containing $X$ have all degree $n$ and there is a finite number of them.
\end{enumerate}
\end{lemma}

\section{Test for strong alt-induced codes}
In this section we propose an algorithm to test for regular strong alt-induced codes.
For this, we need more an auxiliary proposition. 

\begin{proposition} \label{L:Y-uZ}
If $Z = XY$ is a strong alt-induced code, then $X = Zy^{-1}$ and $Y = x^{-1}Z$ for all $x \in X, y \in Y$. Therefore, for all $w \in Z$ there exists $u \in {\rm Pref} (w)$ such that $Y = u^{-1}Z$.
\end{proposition}

\begin{proof}
	Suppose that $Z = XY$ is a strong alt-induced code. Then, by Lemma~\ref{L:siaCode-Char}, $X$ is a prefix code and $Y$ is a suffix code.
	Since $X$ is a prefix code, it follows that $x^{-1}Z = x^{-1}(XY) = (x^{-1}X)Y = \{\varepsilon\}Y = Y$, for all $x \in X$. Hence, $Y = u^{-1}Z$ with $u \in {\rm Pref} (w)$, for all $w \in Z$.
Similarly, we also have $X = Zy^{-1}$, for all $y \in Y$.  \qed 
\end{proof}

Note that, $X$ is a prefix (suffix) code if and only if $X^{-1}X =\{\varepsilon\}$ (resp. $XX^{-1} =\{\varepsilon\}$). By this, and from Lemma~\ref{L:siaCode-Char} and Proposition~\ref{L:Y-uZ} we can exhibit the following algorithm for testing whether a given regular code is strong alt-induced or not.
As usual, when $Z$ is a regular set of $A^+$, $\min Z$ denotes the minimal wordlength of $Z$.

\medskip
{\noindent\bf Algorithm RSIC} {\it (A test for regular strong alt-induced codes)}

{\it Input:} A regular code $Z$ over $A$.

{\it Output:} $Z$ is a strong alt-induced code or not.
 

	1. Choose $w \in Z$, such that $|w| = \min Z$;
		Set $P = {\rm Pref}(w)\setminus \{\varepsilon\}$.

\leftskip 0cm 
	2. While $P \neq \emptyset$ do \{

\leftskip0.8cm 
	   Take $u \in P$; Set $Y = u^{-1}Z$;

	   If $YY^{-1} = \varepsilon$ then \{

\leftskip1.2cm 
	   		Take $y \in Y$; Set $X = Zy^{-1}$;

			If ($X^{-1}X = \varepsilon$) and ($Z = XY$) then goto Step 4;

\leftskip0.8cm 
	   \} 

	   $P := P \setminus \{u\}$;

\leftskip0.45cm 
	\} 

\leftskip 0cm 
	3. $Z$ is not a strong alt-induced code; STOP.

	4. $Z$ is a strong alt-induced code; STOP.

\medskip
Let us take some examples.

\begin{example}
	Consider the set $Z = \{a^2b, a^2ba, ba^2b, ba^2ba\}$ over $A = \{a, b\}$. By Algorithm~RSIC, we have:

	1. Choose $w = a^2b \in Z$ with $|w| = 3 = \min Z$.

\leftskip0.45cm 
		Set $P = {\rm Pref}(w)\setminus \{\varepsilon\} = \{a, aa\}$.

\leftskip 0cm 
	2. While $P \neq \emptyset$ do \{

\leftskip0.8cm 
	   2.1. Take $a \in P$; Set $Y = a^{-1}Z = \{ab, aba\}$.

\leftskip1.5cm 
	   Since $YY^{-1} = \varepsilon$, which implies $X = Zy^{-1}= \emptyset$.

\leftskip0.8cm 
	   $P := P \setminus \{a\} = \{aa\} \neq \emptyset$.

	   2.2. Take $aa \in P$; Set $Y = (aa)^{-1}Z = \{b, ba\}$.

\leftskip1.5cm 
	   Since $YY^{-1} = \varepsilon$, which implies $X = Zb^{-1}= \{a^2, ba^2\}$.

	   Because $X^{-1}X = \varepsilon$ and $Z = XY$, we goto Step 4.
	   
\leftskip0.45cm 
	\} 

\leftskip 0cm 

	4. $Z = \{a^2, ba^2\}\{b, ba\}$ is a strong alt-induced code, and the algorithm ends.
\end{example}

\begin{example}
	Consider the set $Z = \{a^nb^2, a^nb^2ab \ | \ n \geq 1\}$ over $A = \{a, b\}$. By Algorithm~RSIC, we have:

	1. Choose $w = ab^2 \in Z$ with $|w| = 3 = \min Z$.

\leftskip0.45cm 
		Set $P = {\rm Pref}(w)\setminus \{\varepsilon\} = \{a, ab\}$.

\leftskip 0cm 
	2. While $P \neq \emptyset$ do \{

\leftskip0.8cm 
	   2.1. Take $a \in P$; Set $Y = a^{-1}Z = \{a^{n-1}b^2, a^{n-1}b^2ab \ | \ n \geq 1\}$.

\leftskip1.5cm 
	   Since $YY^{-1} \neq \varepsilon$, which implies 
	   $P := P \setminus \{a\} = \{ab\} \neq \emptyset$.

\leftskip0.8cm 
	   2.2. Take $ab \in P$; Set $Y = (ab)^{-1}Z = \{b, bab\}$.

\leftskip1.5cm 
	   Since $YY^{-1} \neq \varepsilon$, it follows that $P = \emptyset\}$.
	   
\leftskip0.45cm 
	\} 

\leftskip 0cm 

	3. $Z$ is not a strong alt-induced code, and the algorithm ends.
\end{example}

\begin{example}
	Consider the set $Z = \{b^na^2b^ma \ | \ n, m \geq 1\}$ over $A = \{a, b\}$. By Algorithm~RSIC, we have:

	1. Choose $w = ba^2ba \in Z$ with $|w| = 5 = \min Z$.

\leftskip0.45cm 
		Set $P = {\rm Pref}(w)\setminus \{\varepsilon\} = \{b, ba, ba^2, ba^2b\}$.

\leftskip 0cm 
	2. While $P \neq \emptyset$ do \{

\leftskip0.8cm 
	   2.1. Take $b \in P$; Set $Y = b^{-1}Z = \{b^{n-1}a^2b^ma \ | \ n, m \geq 1\}$.

\leftskip1.5cm 
	   Since $YY^{-1} \neq \varepsilon$, which implies 
	   $P := P \setminus \{b\} = \{ba, ba^2, ba^2b\} \neq \emptyset$.

\leftskip0.8cm 
	   2.2. Take $ba \in P$; Set $Y = (ba)^{-1}Z = \{ab^ma \ | \ m \geq 1\}$.

\leftskip1.5cm 
	   Since $YY^{-1} = \varepsilon$, which implies $X = Z(aba)^{-1}= \{b^na \ | \ n \geq 1\}$.

	   Because $X^{-1}X = \varepsilon$ and $Z = XY$, we goto Step 4.
	   
\leftskip0.45cm 
	\} 

\leftskip 0cm 

	4. $Z = \{b^na \ | \ n \geq 1\} \{ab^ma \ | \ m \geq 1\}$ is a strong alt-induced code, and the algorithm ends.
\end{example}

\begin{theorem} \label{T:oRSIC}
	Given a regular code $Z$ over $A$, we can determine whether $Z$ is a strong alt-induced code or not in $O(m^3)$ worst-case time, 
where $m$ is the number of states in the deterministic finite automaton recognizing $Z$.
\end{theorem}

\begin{proof}
Suppose that $\mathcal A$ is a minimal deterministic finite automaton recognizing $Z$, with $m$ is the number of states in $\mathcal A$.
Then, any word in $Z$ of minimal length which has the length less than or equal to $m$.
   
	In Step 1, we can choose $w \in Z$ such that $|w| = \min Z \leq m$, and the set $P$ has at most $m-1$ words. The best-case for the algorithm is when $P = \emptyset$, which takes $O(1)$ to perform the task.
In Step 2, for each word $u$ in $P$, it takes $O(m^2)$ worst-case time to perform for $YY^{-1}$ or $X^{-1}X$, and $O(m)$ worst-case time in finding $Y$ or $X$. Thus, the total running time for determining, whether $Z$ is a strong alt-induced code or not, is $O(m).O(m^2) = O(m^3)$ in the worst-case. \qed
\end{proof}

\section{Maximal strong alt-induced codes}
As usual, a language $Z$ is a {\it prefix (suffix, bifix, p-infix, s-infix, infix, p-subinfix, s-subinfix, subinfix, hyper) alt-induced (strong alt-induced) code} if it is an alt-induced (a strong alt-induced) code as well as a prefix (resp., suffix, bifix, p-infix, s-infix, infix, p-subinfix, s-subinfix, subinfix, hyper) code.

As well-known, the product of two alt-induced (strong alt-induced) codes is not, in general, a strong alt-induced.
For special subclasses of alt-induced (strong alt-induced) codes, we have however

\begin{proposition} \label{P:icode-Clocat} The following holds true
\begin{enumerate}
\renewcommand{\labelenumi} {\rm(\roman{enumi})} 
	\item The product of two prefix (suffix, bifix, p-infix, s-infix, infix, p-subinfix, s-subinfix, subinfix, hyper) alt-induced (strong alt-induced) codes is a prefix (resp., suffix, bifix, p-infix, s-infix, infix, p-subinfix, s-subinfix, subinfix, hyper) alt-induced (strong alt-induced) code.
	\item The product of two maximal prefix (suffix, bifix) alt-induced (strong alt-induced) codes is a maximal prefix (resp., suffix, bifix) alt-induced (strong alt-induced) code.
\end{enumerate}
\end{proposition}

\begin{proof}
	(i) As known in~\cite{HfV2017}, the product of two prefix (suffix, bifix) alt-induced codes is a prefix (resp., suffix, bifix) alt-induced code. Now, we treat only the case of p-infix strong alt-induced codes. For the other cases the arguments are similar. Let $Z=XY$ and $Z'=X'Y'$ be p-infix strong alt-induced codes. Then, by Lemma~\ref{L:Cph-Clocat}, $ZZ'$ is a p-infix code. On the other hand, since $Z$ and $Z'$ are strong alt-induced codes,
by Lemma~\ref{L:siCpsb-Char}(i), $X, X'$ are prefix codes and $Y, Y'$ are suffix codes. Therefore, again by Lemma~\ref{L:siCpsb-Char}(i), $ZZ'$ is a strong alt-induced code because $ZX'$ is a prefix code and $Y'$ is a suffix codes. Hence, $ZZ'$ is a p-infix strong alt-induced code.

	(ii) Suppose that $Z = XY$ and $Z' = X'Y'$ are maximal prefix strong alt-induced codes. By (i) of the proposition, $ZZ'$ is a prefix strong alt-induced code.
	 On the other hand, since $Z$ and $Z'$ are maximal strong alt-induced codes,
by Lemma~\ref{L:siCpsb-Char}(i), $X, X'$ are maximal prefix codes and $Y, Y'$ are both maximal prefix codes and bifix codes. Therefore, again by Lemma~\ref{L:siCpsb-Char}(i), $ZZ'$ is a maximal prefix strong alt-induced code because $XYX'$ is a maximal prefix code and $Y'$ is both a maximal prefix code and a bifix code. For the remaining cases the arguments are similar.  \qed
\end{proof}

We now show that the embedding problem for the
classes of prefix (suffix, bifix) strong alt-induced codes has a positive solution in the regular case.

\begin{theorem} \label{T:EP-Cr.sipsb}
	Every regular prefix (suffix, bifix) strong alt-induced code is contained in a maximal regular maximal prefix (resp., suffix, bifix) strong alt-induced code.
\end{theorem}

\begin{proof}
	We deal with only the case of prefix strong alt-induced codes. For the other cases the argument is similar.
Suppose $Z$ is a regular prefix strong alt-induced code over $A$ such that $Z = XY$ with $\emptyset \neq X, Y \subseteq A^+$.

	Firstly, we show that $X$ and $Y$ are regular. Indeed, recall that, for any language $Z$ over $A$, the syntactic congruence of $Z$, denote by $\cong_Z$, is defined as follows:
	$${\rm For~}  u, v\in A^*: u \cong_Z v \ {\rm if~ and~ only~ if~}  \forall x, y\in A^*: xuy \in Z \Leftrightarrow xvy \in Z.$$
Then, it is easy to verify that, for any word $x \in A^*$, $x^{-1}Z$ is a union of equivalence classes of $\cong_Z$. Now, since $Z$ is regular, $\cong_Z$ has finite index, i.e. there exists only a finite number of equivalence classes according to the congruence $\cong_Z$.
On the other hand, because $Z$ is a strong alt-induced code, by Definition~\ref{D:sicode}, are verified the following equalities:
$$X = \bigcup_{y \in Y} Zy^{-1}, \ Y = \bigcup_{x \in X} x^{-1}Z.$$
Thus, by the above, $X$ and $Y$ are also unions of a finite number of equivalence classes of $\cong_Z$. Hence, $X$ and $Y$ are regular.

	Next, since $Z$ is a prefix strong alt-induced code, by Lemma~\ref{L:siCpsb-Char}(i), $X$ is a prefix code and $Y$ is a bifix code.
Therefore, by Lemma~\ref{L:EP-Cpsb}(i), there exist $M_X$ is a regular maximal prefix code and $M_Y$ is a regular maximal bifix code such that $X \subseteq M_X$ and $Y \subseteq M_Y$. Since $M_Y$ is a regular code, by Lemma~\ref{L:ThinReg}, $M_Y$ is a thin code. Thus, by Lemma~\ref{L:ThinMb}, $M_Y$ is both a regular maximal prefix code and a regular maximal suffix code. Hence, again by Lemma~\ref{L:siCpsb-Char}(i), $M_XM_Y$ is a regular maximal prefix strong alt-induced code which contains $Z$. \qed
\end{proof}

\begin{example}
	Let $X = (aa)^*b$, $Y = ab^*ab$ over $A = \{a, b\}$, and $Z=XY$. Then, $X$ is a regular prefix code and $Y$ is a regular bifix code. By Lemma~\ref{L:siCpsb-Char}(i), $Z$ is a regular prefix strong alt-induced code. It is not difficult to check that $M_X = a^*b$ is a regular maximal prefix code and $M_Y = ba + bb + ab^*a(a+b)$ is a regular maximal bifix code.
Thus, $M_XM_Y$ is a regular maximal prefix strong alt-induced code which contains $Z$.
\end{example}

\begin{example}
	Let $Z = (b + a(bb)^*a) (a + ab)$ over $A = \{a, b\}$. Then, by Lemma~\ref{L:siCpsb-Char}(ii), $Z$ is a regular suffix strong alt-induced code because $b + a(bb)^*a$ is a regular bifix code and $a + ab$ is a finite suffix code. It is easy to check that $b + ab^*a$ is a regular maximal bifix code and $a + ab + bb$ is a finite maximal suffix code. Hence, $(b + ab^*a)(a + ab + bb)$ is a regular maximal suffix strong alt-induced code which contains $Z$.
\end{example}

\begin{example}
	Over $A = \{a, b, c\}$, let $X = a+c$, $Y = c(a+b)^*c$ and $Z=XY$. Then, $X$ and $Y$ are regular bifix codes. By Lemma~\ref{L:siCpsb-Char}(iii), $Z$ is a regular bifix strong alt-induced code. It is easy to verify that $M_X = A$ is a finite maximal bifix code and $M_Y = a + b + c(a+b)^*c$ is a regular maximal bifix code.
Thus, $M_XM_Y$ is a regular maximal bifix strong alt-induced code which contains $Z$.
\end{example}

It is not true in general that any finite bifix code is contained in a finite maximal
one. Indeed, if we consider $X = \{a, bb\}$, then any maximal bifix code containing $X$ is
infinite, since no word of $ba^*$ can be added to $X$.
Thus, according to Lemma~\ref{L:siCpsb-Char} and Lemma~\ref{L:EP-Cpsb}, it is not true in general that any finite prefix (suffix, bifix) strong alt-induced code is contained in a finite maximal one.

The following result shows that a finite prefix (suffix, bifix) strong alt-induced code is always contained in a regular maximal prefix (suffix, bifix) strong alt-induced code.

	
\begin{theorem} \label{T:EP-Cf.sips}
	Let $Z=XY$ be a finite prefix (suffix, bifix) strong alt-induced code. Then,
\begin{enumerate}
\renewcommand{\labelenumi} {\rm(\roman{enumi})} 
	\item $Z$ is contained in a regular maximal prefix (resp. suffix, bifix) strong alt-induced code.
	\item In particular, $Z$ is contained in a finite maximal prefix (resp. suffix, bifix) strong alt-induced code, if moreover $Y$ is a nice bifix code (resp. $X$ is a nice bifix code, $X$ and $Y$ are nice bifix codes). 
\end{enumerate}
\end{theorem}

\begin{proof}
	(i) It follows immediately from Theorem~\ref{T:EP-Cr.sipsb}.

	(ii) Suppose $Z=XY$ is a finite prefix strong alt-induced code and $Y$ is a nice bifix code. Then, by the assumption and Lemma~\ref{L:siCpsb-Char}(i), $X$ is a finite prefix code and $Y$ is a nice bifix code.
Therefore, on one hand, by Lemma~\ref{L:EP-Cpsb}, there exists $M_X$ is a finite maximal prefix code such that $X \subseteq M_X$. On the other hand, by Lemma~\ref{L:EP-Cfb}, there exists $M_Y$ is a finite maximal bifix code  which contains $Y$.
Hence, again by Lemma~\ref{L:siCpsb-Char}(i), $M_XM_Y$ is a finite maximal prefix strong alt-induced code which contains $Z$.
For the other cases the argument is similar. \qed
\end{proof}

\begin{example}
	Consider the sets $Z = \{a, ba\}\{a, bb\}$ and $Z' = \{a, ba\}\{a\}$ over $A = \{a, b\}$. Then, by Lemma~\ref{L:siCpsb-Char}(i), $Z$ and $Z'$ are finite prefix strong alt-induced codes because $\{a, ba\}$ is a prefix code, and $\{a, bb\}$ and $\{a\}$ are bifix codes.
It is easy to see that $\{a, ba, bb\}$ is a maximal prefix code, and $\{a, ba^nb \ | \ n\geq 0\}$ is a regular maximal bifix code. Therefore, $\{a, ba, bb\}\{a, ba^nb \ | \ n\geq 0\}$ is a regular maximal prefix strong alt-induced code which contains $Z$.
On the another hand, since $\{a\}$ is a nice bifix code, it follows that $\{a\}$ is included in the finite maximal bifix code $\{a, b\}$. Thus, $\{a, ba, bb\}\{a, b\}$ is a finite maximal prefix strong alt-induced code which contains $Z'$.
\end{example}

\begin{example}
	We consider the sets $X = \{a, b\}$ and $Y = \{aa, aba, bba, b, c\}$ over $A = \{a, b, c\}$. Then, $X$ is a finite bifix code and $Y$ is a finite suffix code. 
By Lemma~\ref{L:siCpsb-Char}(ii), $Z =  XY$ is a finite suffix strong alt-induced code. Clearly, $M_X = A$ is a finite maximal bifix code and $M_Y = \{aa, ca, aba, bba, cba, b, c\}$ is a finite maximal suffix code.
Thus, $M_XM_Y$ is a finite maximal suffix strong alt-induced code which contains $Z$.
\end{example}

\begin{example}
	Over $A = \{a, b\}$, let $X = a + bbb$, $Y = b + aa$ and $Z=XY$. Then, $X$ and $Y$ are finite bifix codes. By Lemma~\ref{L:siCpsb-Char}(iii), $Z$ is a finite bifix strong alt-induced code. It is not difficult to verify that $M_X = a + ba^*ba^*b$ and $M_Y = b + ab^*a$ which are regular maximal bifix codes.
Thus, $M_XM_Y$ is a regular maximal bifix strong alt-induced code which contains $Z$.
\end{example}

\section{Conclusions}
The purpose of this note was to deal with the development of the class of alt-induced codes.
Closure under concatenation property for several subclasses of alt-induced codes was proposed (Proposition~\ref{P:icode-Clocat}).
An algorithm, with a cubic time complexity (Theorem~\ref{T:oRSIC}), to test whether a regular code is strong alt-induced or not was established (Algorithms~RSIC).
Especially, the embedding problem for the classes of prefix (suffix, bifix) strong alt-induced codes has a solution in both the finite and regular case (Theorem~\ref{T:EP-Cr.sipsb}, Theorem~\ref{T:EP-Cf.sips}).

In future works, we hope we can solve the embedding problem for the classes of strong alt-induced codes and many interesting problems for alt-induced codes as well as their subclasses.

\section*{Acknowledgment}
	The author would like to thank the colleagues in Seminar Mathematical Foundation of Computer Science at Institute of Mathematics, Vietnam  Academy of Science and Technology for attention to the work. Especially, the author expresses her sincere thanks to Prof. Do Long Van and Dr. Kieu Van Hung for their useful discussions. 


\end{document}